\documentclass[journal]{IEEEtran}
\usepackage[dvips]{graphicx}
\usepackage{amsmath, amsthm, amssymb, amsfonts}
\usepackage{cases, multirow, cite}
\usepackage{srcltx,color}
\usepackage{subfigure}

\long\def\comment#1{}

\newtheorem{corollary}{Corollary}

\newtheorem{remark}{Remark}

\def\figref#1{Fig.~\ref{#1}}

\begin{document}

\title{Joint Optimization of Area Spectral Efficiency and Delay Over PPP Interfered Ad-hoc Networks}

\author{
   \IEEEauthorblockN{Young Jin Chun, Aymen Omri, and Mazen O. Hasna}
   \thanks{Y. J. Chun is with the Wireless Communications Laboratory, ECIT Institute, Queens University Belfast, United Kingdom
   (Email: ychun@ieee.org).}
   \thanks{A. Omri and M. O. Hasna are with Department of Electrical Engineering, Qatar University, Doha, Qatar (Email: {omriaymen,hasna}@qu.edu.qa).} 
   \thanks{This work was supported in part by the Engineering and Physical Sciences Research Council (EPSRC) under Grant References EP/H044191/1 and EP/L026074/1, and in part by the NPRP Grant 4-1119-2-427 from the Qatar National Research Fund (a member of Qatar Foundation). The statements made herein are solely the responsibility of the authors.}
 }

\maketitle
\begin{abstract}
To evaluate the performance of the co-channel transmission based communication, 
we propose a new metric for area spectral efficiency (ASE) of interference limited Ad-hoc network by assuming that the nodes are randomly distributed according to a Poisson point processes (PPP).
We introduce a utility function, $U = ASE/\text{delay}$ and derive the optimal ALOHA transmission probability $p$ and the SIR threshold $\tau$
that jointly maximize the ASE and minimize the local delay. Finally numerical results has been conducted to confirm that the joint optimization based on the $U$ metric achieves a significant performance gain compared to conventional systems.
\end{abstract}

\vspace{-0.1in}

\section{Introduction}
\IEEEPARstart{S}{}
upporting high-data-rate transmissions over limited radio spectrum with minimum amount of power consumption is the main goal of future wireless communications systems \cite{Ref12}, \cite{Ref32}, \cite{Ref_ext_002}. To reach higher spectral efficiencies, wireless network technologies need to collaborate and construct a seamless interconnection between multiple tiers, 
such as macro and small cells. Such interconnection between multiple tiers has been shown to improve spectral efficiency and coverage \cite{Ref32}. 
However, it increases the in-band interference, leading to a drastic degradation in the overall networks performance. 
Therefore, the optimal network design for future wireless communication systems should address the effect of interference within the area spectral efficiency metric. In addition, the local delay should be re-evaluated based on the interference model. To evaluate the performance of the new network design, a new area spectral efficiency metric is needed.
 
Most conventional performance metrics in wireless communication systems focus on the quantification of either link reliability or spectral efficiency \cite{RefA1, RefA2, RefH1, Ref22}. The link reliability is usually quantified in terms of outage probability or average error rate \cite{RefA1}, \cite{RefA2}, while the ergodic capacity addresses the spectral efficiency of wireless links, defined as the maximum achievable average spectrum efficiency \cite{RefH1}, \cite{Ref22}. 
However, those metrics are still short of taking into account the spatial effect of wireless transmissions.
In \cite{RW1}, a general area spectral efficiency (GASE) is defined to be the average data rate per unit bandwidth per unit area supported by a base station. In \cite{RW1}, the authors introduced the concept of affected area, based on which they evaluated the spectral and power efficiency of different transmission scenarios. As such, area spectral efficiency metrics are proposed to evaluate conventional wireless communication systems. These metrics, can not be used to evaluate the ASE in the presence of co-channel interference, where the interference problem has major impact on the system performance, and it is not included in the definition of the previous ASE metric.

In this paper, we propose a new metric for ASE of interference limited Ad-hoc networks 
by assuming that the nodes are randomly distributed according to a Poisson point process (PPP).
In particular:
\begin{itemize}
	\item The ASE of PPP interfered Ad-hoc networks is defined by using the concept of affected area
	and derived by using stochastic geometry. 	
  \item A utility function, $U = ASE/\text{delay}$, is introduced to jointly optimize the ASE and the delay,
  \item The optimal ALOHA transmission probability $p$ and the $SIR$ threshold $\tau$ that maximizes the utility $U$ is determined.
\item The significant performance gain acheived of the joint optimization based on $U$ when compared to conventional systems is numerically proved.
\end{itemize}
The remainder of this paper is organized as follows: Section II describes the system model. The system measure analysis is detailed in section III. Section IV presents the joint optimization of the ASE and the delay. The numerical results are presented and interpreted in Section V. Finally, conclusions are drawn in section VI.

\section{System Model}
We consider an ad-hoc network where multiple transmitters simultaneously transmit during each time slot and 
the spatial location of the interfering node is modeled as a PPP denoted as $\Phi$ with intensity $\lambda$. Specifically, we focus on a point-to-point link within the ad-hoc network 
where the destination node $d$ is located at the origin $d = (0,0)$, the source node is $s$, the interfering nodes are denoted by $x_n$ with node index $n \geq 0$, \textit{i.e.}, $x_n \in \Phi$. We let $s$, $d$, $x_n$ denote both the nodes and their coordinates. For MAC protocol, we use ALOHA with transmit probability $p$; each transmission occurs during one time slot where the nodes attempt to access a certain slot with probability $p$. If the transmission fails, the node attempts to re-transmit. The nodes transmit with power $P_t$, the distance between nodes $i$ and $j$ is denoted by $d_{ij} = ||i-j||$, and the path loss function between any two nodes is given by 
$||i-j||^{-\alpha} = d_{ij}^{-\alpha}$, where $\alpha > 2$ is the path loss exponent.
The channel coefficient between node $i$ and the destination $d$ is denoted by $h_{i}$.
We assume an independent and identically distributed (i.i.d.) Rayleigh fading variable with mean one, 
resulting in an exponential power distribution with mean one.
To simplify the analysis, we consider an interference limited environment 
where the interference is significantly larger than the noise, so that the noise will be ignored through out the analysis. 

\section{System Measures Analysis}
\subsection{Distribution of the SIR} 
Let $\Phi_k$ denote the set of active interfering nodes during time slot $k$, \textit{i.e.}, $\Phi_k \in \Phi$.
The aggregate interference at the destination $d$ in time slot $k$ is given by
\begin{equation}
	\begin{split}
	I_k = P_t \sum_{x \in \Phi \backslash \{s\}} h_{x} ||x||^{-\alpha} \textbf{1}\left(x \in \Phi_k\right),
	\end{split}
	\label{eq_01}
\end{equation}
where $\textbf{1}(\cdot)$ is the indicator function and the intended signal from $s$ is excluded from (\ref{eq_01}). 
Then, the signal-to-interference ratio (SIR), denoted by $\Gamma_k$, is
$\Gamma_k = {P_t h_{s} d_{sd}^{-\alpha}}/{I_k} = {h_{s} d_{sd}^{-\alpha}}/{I_k^{'}}$,
where $I_k^{'} \triangleq I_k/P_t$ for notation simplicity. 

Let $\mathcal{C}_{\Phi}$ denote the successful transmission event conditioned on the PPP $\Phi$. Given the SIR threshold $\tau$, 
the probability of successful transmission conditioned on $\Phi$ can be evaluated as  
\begin{equation}
	\begin{split}
	\mathbb{P}\left( \mathcal{C}_{\Phi} \right) &= p \mathbb{P}\left( \Gamma_k > \tau | {\Phi} \right)
	= p \mathbb{P}\left[ h_{s} >  \tau d_{sd}^{\alpha} I_k^{'} | {\Phi} \right]\\
	&= p \mathbb{E}\left[ 
	\prod_{x \in \Phi \backslash \{s\}}
	\mathrm{e}^{ - \tau d_{sd}^{\alpha} h_{x} ||x||^{-\alpha} \textbf{1}\left(x \in \Phi_k\right) }
	\right]\\
	&= p \prod_{x \in \Phi \backslash \{s\}}
	\left[ 
	\frac{p}{1+ \tau d_{sd}^{\alpha} ||x||^{-\alpha}} + 1-p
	\right],
	 \end{split}
	\label{eq_03}
\end{equation}
where the first equality follows because a transmission occurs with probability $p$,
the third equality is derived using the distribution of 
$h_{k, s}$, \textit{i.e.}, $\mathbb{P}(h_{s} > x) = \exp(-x)$, 
and the expectation in the last equality is taken with regard to the Rayleigh faded coefficient $h_{x}$ and 
the ALOHA protocol $\Phi_k$.

The distribution of SIR can be evaluated by averaging $\mathbb{P}\left( \Gamma_k > \tau | {\Phi} \right)$ over all possible $\Phi$ as follows.
\begin{equation}
	\begin{split}
	&\mathbb{P}\left( \Gamma_k > \tau \right) = \mathbb{E}_{\Phi}\left[ \mathbb{P}\left( \Gamma_k > \tau | {\Phi} \right) \right]\\
	&= \exp\left(
	-\lambda p \int_{R^2} \left[ 
	1 - \frac{1}{1+ \tau d_{sd}^{\alpha} ||x||^{-\alpha}}
	\right] \mathrm{d}x
	\right) \\
	&= \exp\left(
	-2 \pi \lambda p \int_{\rho = 0}^{\infty}  \frac{\rho}{\left( \tau d_{sd}^{\alpha} \right)^{-1} \rho^{\alpha} + 1}
	 \mathrm{d}\rho
	\right), 
	 \end{split}
	\label{eq_04}
\end{equation}
where we applied the probability generating functional (PGFL) of the PPP \cite{ref-02} in the second step and 
changed the Cartesian coordinates to Polar coordinates in the third step.
By applying the integral relation in (\ref{eq_04}),
	\begin{equation}
	\begin{split}
\int_{0}^{\infty} \frac{x^{\mu - 1}}{1 + q x^{\nu}}\mathrm{d}x = \frac{1}{\mu} q^{-\frac{\mu}{\nu}} C\left(\frac{\mu}{\nu}\right),
\quad C(\delta) = \frac{1}{\mathrm{sinc}(\delta)},
		\end{split}
		\label{neo.eq_1}
	\end{equation}
we can simplify the distribution of SIR as follows
\begin{equation}
	\begin{split}
	\mathbb{P}\left( \Gamma_k > \tau \right) = \exp\left(
	-\lambda \pi d_{sd}^{2} C(\delta) p \tau^{\delta}
	\right),\quad \delta = \frac{2}{\alpha}.
	 \end{split}
	\label{neo.eq_2}
\end{equation}

\subsection{Area Spectral Efficiency} 

In this subsection, we evaluate the ASE of a PPP interfered ad-hoc network using the definition in \cite{RW1} as follows
\begin{equation}
	\begin{split}
	A_e &\triangleq  \frac{C}{\Lambda},
	 \end{split}
	\label{neo.eq_3}
\end{equation}
where $A_e$, $C$, and $\Lambda$ indicate the ASE, the transmission capacity, and the affected area, respectively.
First, the capacity is evaluated as
\begin{equation}
	\begin{split}
	C &= p \mathbb{E}\left[ \ln(1 + \Gamma)\right]  = p \int_{0}^{\infty} \mathbb{P}\left[ \ln(1 + \Gamma) > t \right] \mathrm{d}t \\
	&= p \int_{0}^{\infty} \mathbb{P}\left[ \Gamma > \mathrm{e}^t-1 \right] \mathrm{d}t\\
	&= p \int_{0}^{\infty} \exp\left[ - A^{'} p \left( \mathrm{e}^t - 1 \right)^{\delta} \right] \mathrm{d}t,
	 \end{split}
	\label{neo.eq_4}
\end{equation}
where the last step follows from (\ref{neo.eq_2}), and $A^{'} \triangleq \lambda \pi d_{sd}^{2} C(\delta)$.

Next, we calculate the affected area by finding the spatial range $d_0$ that 
achieves a minimum link success probability $p_s$
\begin{equation}
	\begin{split}
	\underset{d_{sd} \leq d_0}{\text{Find}} ~ d_{0} ~ \text{such that}~ \mathbb{P}\left( \Gamma \geq \tau \right) \geq p_s.
	 \end{split}
	\label{neo.eq_5}
\end{equation}
Then, the maximum range is given by
\begin{equation} 
\begin{split}
	&\mathbb{P}\left( \Gamma \geq \tau \right) \geq p_s
	\Leftrightarrow ~d_{sd}^{2} \leq -\frac{\ln(p_s)}{\lambda \pi C(\delta) p \tau^{\delta}}\\
	\Rightarrow & d_0 \triangleq \sqrt{\frac{|\ln(p_s)|}{\lambda \pi C(\delta) p \tau^{\delta}}},  ~\quad 0 \leq p_s \leq 1,
\end{split}
	\label{eq_08}
\end{equation}
and the affected area is given as follows
\begin{equation}
\begin{split}
	\Lambda & = \pi d_0^2 = \frac{|\ln(p_s)|}{\lambda C(\delta) p \tau^{\delta}}.
\end{split}
	\label{eq_09}
\end{equation}

Therefore, the ASE of a PPP interfered ad-hoc networks is
\begin{equation}
	\begin{split}
	A_e = \frac{\lambda C(\delta) p^2 \tau^{\delta}}{|\ln(p_s)|} \int_{0}^{\infty} \exp\left[ - A^{'} p \left( \mathrm{e}^t - 1 \right)^{\delta} \right] \mathrm{d}t,
	 \end{split}
	\label{neo.eq_6}
\end{equation}
where $A^{'} = \lambda \pi d_{sd}^{2} C(\delta)$, $C(\delta) = \frac{1}{\mathrm{sinc}(\delta)}$, and $\delta = \frac{2}{\alpha}$.

\subsection{Mean Local Delay}
Each node attempts to re-transmit if the transmission during the previous time slot failed. 
In general, the success transmission events in different time slots are dependent due to interference correlation. However, if we consider the conditional success event for a given $\Phi$, the randomness originates only from the channel fading coefficient and the ALOHA protocol that are independent for each time slot. Therefore, the success event in different time slots given $\Phi$ are independent with probability $P\left( \mathcal{C}_{\Phi} \right)$ in (\ref{eq_03})  
and the local delay given $\Phi$, which is denoted as $\Delta_{\Phi}$ and defined as the number of time slots required till a successful transmission occurs, 
is a geometric random variable as given below
\begin{equation}
	\begin{split}
		\mathbb{P}\left( \Delta_{\Phi} = k \right) = \left( 1 - \mathbb{P}\left( \mathcal{C}_{\Phi} \right) \right)^{k-1} 
		\mathbb{P}\left( \mathcal{C}_{\Phi} \right).
	 \end{split}
	\label{neo.eq_7}
\end{equation}
The mean local delay $D(p)$ for every possible $\Phi$ is given by
\begin{equation}
	\begin{split}
	D(p) &\triangleq \mathbb{E}_{\Phi}\left[ \mathbb{E}\left( \Delta_{\Phi} \right) \right] = 
	\mathbb{E}_{\Phi}\left[ \frac{1}{\mathbb{P}\left( \mathcal{C}_{\Phi} \right)} \right]\\
	&= \frac{1}{p} \exp\left( \frac{A^{'} p \tau^{\delta}}{( 1-p)^{1-\delta}} \right),
	 \end{split}
	\label{neo.eq_8}
\end{equation}
where the second and third equality follows from using the mean of geometric random variable and (\ref{eq_03}), respectively. 
Detailed proof of (\ref{neo.eq_8}) is provided in \cite{MAC-Ref01} and \cite{Ref_ext_001}.§

\section{Joint Optimization of the ASE and Delay}
Let us define a utility function $U \triangleq \frac{A_e}{D(p)}$ for PPP interfered ad-hoc networks. By maximizing $U$, we can jointly maximize the area spectral efficiency and minimize the local delay. First, we determine the optimal SIR threshold $\tau$ that maximizes $U$ for a given transmission probability $p$ as follows
\begin{equation}
	\begin{split}
	\underset{\tau \geq 0}{\text{Find}} ~ \tau^{\ast} ~ \text{that maximizes}~ U, ~\textit{i.e.,~} \frac{\partial U}{\partial \tau} = 0.
	 \end{split}
	\label{neo.eq_10}
\end{equation}
Then, we find the optimal ALOHA transmission probability $p$ that maximizes $U$ for a given SIR threshold as follows
\begin{equation}
	\begin{split}
	\underset{0 \leq p \leq 1}{\text{Find}} ~ p^{\ast} ~ \text{that maximizes}~ U, ~\textit{i.e.,~} \frac{\partial U}{\partial p} = 0.
	 \end{split}
	\label{neo.eq_9}
\end{equation}

\begin{corollary}
The optimal SIR threshold $\tau^{\ast}$ that maximizes the utility $U$ in (\ref{neo.eq_10}) for a given $p$ is 
\begin{equation}
	\begin{split}
	\tau^{\ast} = q (1-p)^{\frac{1-\delta}{\delta}}, 
    \quad q = \left( {A^{'} p}\right)^{-1/\delta}.
	 \end{split}
	\label{neo.eq_12}
\end{equation}
The optimal probability $p^{\ast}$ that maximizes $U$ for a given threshold is the solution of the following condition 
  \begin{equation}
    \begin{split}
    	\left[\frac{3}{A^{'} p^{\ast}} - \frac{\tau^{\delta} \left( 1 - p^{\ast} \delta \right)}{\left( 1 - p^{\ast}\right)^{2-\delta}} \right] 
        = \frac{\psi_1(p^{\ast})}{\psi_0(p^{\ast})},
    \end{split}
	    \label{gene.neoeq_03}
  \end{equation}
where $\psi_n(p)$ denotes the integral equation
  \begin{equation}
	\psi_n(p) \triangleq \int_{0}^{\infty} \left( \mathrm{e}^t-1\right)^{n \delta} 
    \exp\left( - A^{'} p \left( \mathrm{e}^t-1\right)^{\delta}  \right) \mathrm{d}t,
	\label{gene.neoeq_04}
  \end{equation}
  for non-negative integer $n = 0, 1, 2, \cdots$.
\end{corollary}

\begin{proof}
By using (\ref{neo.eq_6}), (\ref{neo.eq_8}) and (\ref{gene.neoeq_04}), 
the utility $U$ is given as 
   \begin{equation}
	U \triangleq \frac{A_e}{D(p)} 
    = \frac{\lambda C(\delta) p^3 \tau^{\delta}}{|\ln(p_s)|} 
    \exp\left( -\frac{A^{'} p \tau^{\delta}}{( 1-p)^{1-\delta}} \right)
    \psi_0(p).
	    \label{eq.twi_001}
  \end{equation}
The first derivative of $U$ with respect to $\tau$ is obtained as 
   \begin{equation}
    \begin{split}
    \frac{\partial U}{\partial \tau} &= \left( \frac{1}{\tau} -\frac{A^{'} p}{(1-p)^{1-\delta}} \tau^{\delta-1} \right) \delta U.
    \end{split}
	    \label{eq.cor_001}
  \end{equation}
Since $U$ has non-negative value, the optimal $\tau^{\ast}$ achieves (\ref{neo.eq_12}), 
where we substituted $1/q^{\delta} = A^{'} p$ in (\ref{eq.cor_001}). 
The second derivative of $U$ with respect to $\tau$ at $\tau^{\ast}$ is 
   \begin{equation}
    \begin{split}
	\frac{\partial^2 U}{\partial \tau^2}\bigg|_{\tau^{\ast}} &= 
	\left( \frac{1}{\tau} -\frac{A^{'} p}{(1-p)^{1-\delta}} \tau^{\delta-1} \right) \delta \frac{\partial U}{\partial \tau}\\
	&- 
	\left( \frac{1}{\tau^2} + \frac{A^{'} p (\delta-1) }{(1-p)^{1-\delta}} \tau^{\delta-2} \right) \delta U \bigg|_{\tau^{\ast}}\\
	&= -\frac{\delta U}{(\tau^{\ast})^2} \left( 1 -\frac{A^{'} p (1-\delta)}{(1-p)^{1-\delta}} (\tau^{\ast})^{\delta} 
	\right)\bigg|_{\tau^{\ast}}\\
	&= -\frac{\delta U}{(\tau^{\ast})^2} \left( 1 - (1-\delta)\right) = -\frac{\delta^2 U}{(\tau^{\ast})^2},
    \end{split}
	    \label{eq.cor_002}
  \end{equation}
  where we applied $\left. \frac{\partial U}{\partial \tau} \right\vert_{\tau^{\ast}} = 0$ in the second equality and 
  used (\ref{neo.eq_12}) in the last equality. Since the second derivative of $U$ is negative at $\tau^{\ast}$, 
  the SIR threshold (\ref{neo.eq_12}) maximizes the utility. 

	The first derivative of $U$ with respect to $p$ is obtained as
	   \begin{equation}
	    \begin{split}
	    \frac{\partial U}{\partial p} &= \frac{\lambda C(\delta) p^3 \tau^{\delta} A^{'}}{|\ln(p_s)|} 
	    \exp\left( -\frac{A^{'} p \tau^{\delta}}{( 1-p)^{1-\delta}} \right)\\
	    &\times 
	    \left[ \left(\frac{3}{ A^{'} p} - \frac{\tau^{\delta} \left( 1 - p \delta \right)}{\left( 1 - p\right)^{2-\delta}} \right)  
	    \psi_0(p) - \psi_1(p)
	    \right],
	    \end{split}
		    \label{gene.eq_36}
	  \end{equation}
	  where we used (\ref{gene.neoeq_04}), (\ref{eq.twi_001}), and $\frac{\partial \psi_n(p)}{\partial p} = -A^{'} \psi_{n+1}(p)$.
	 Since $\psi_n(p)$ is an integral of non-negative valued function, $\psi_n(p)$ is non-negative and 
	 the optimal $p^{\ast}$ achieves (\ref{gene.neoeq_03}).
	 The second derivative of $U$ with respect to $p$ has negative value at $p^{\ast}$, 
	 which can be proved by using the similar approach as (\ref{eq.cor_002}). Hence, 
	 the transmission probability $p^{\ast}$ maximizes the utility $U$ and this completes the proof. 
\end{proof}

\begin{remark}
The integral $\psi_n(p)$ can be numerically computed by applying change of variable on (\ref{gene.neoeq_04}) three times as follows; 
$1/q^{\delta} = A^{'} p$, $\frac{\mathrm{e}^t-1}{q} \rightarrow x$, and $x^{\delta} \rightarrow y$. 
  \begin{equation}
  	\begin{split}
	\psi_n(p) &= q^{n\delta + 1} \int_{0}^{\infty} \frac{x^{n\delta} \mathrm{e}^{-x^{\delta}}}{1+qx} \mathrm{d}t\\ 
	&= \frac{q^{n\delta + 1}}{\delta} \int_{0}^{\infty} \frac{y^{n+\frac{1}{\delta}-1} \mathrm{e}^{-y}}{1+q y^{\frac{1}{\delta}}} \mathrm{d}t\\ 
	&= \frac{q^{n\delta + 1}}{\delta} G\left( n + \frac{1}{\delta}, \frac{1}{\delta}, q \right),
    \end{split}
	\label{eq.twi_002}
  \end{equation}
where we denote the following integral as $G\left( \mu, \nu, q \right)$  
\begin{equation}
	\begin{split}
	&G\left( \mu, \nu, q \right) \triangleq \int_{0}^{\infty} \frac{t^{\mu - 1} \mathrm{e}^{-t}}{1 + q t^{\nu}}\mathrm{d}t\\
	&= \sum_{n = 0}^{\infty} \frac{(-1)^n}{n! (n+\mu-1)} q^{-\frac{n+\mu-1}{\nu}} C\left( \frac{n+\mu-1}{\nu} \right).
	\end{split}
	\label{neo2.eq_2}
\end{equation}
We replaced $\mathrm{e}^{-x}$ in (\ref{neo2.eq_2}) with its Taylor series form $\mathrm{e}^{-x} = \sum_{n=0}^{\infty}\frac{(-1)^n x^n}{n!}$ 
and applied (\ref{neo.eq_1}) in (\ref{neo2.eq_2}).
\end{remark}

\section{Numerical Results}
In this section, we evaluate and compare the newly developed metrics; $A_e$, $D(p)$, and $U = {A_e}/{D(p)}$.
\figref{fig.num1} compares the utility $U$ versus node density $\lambda$ for several transmission probabilities $p$.
We used (\ref{neo.eq_6}) and (\ref{neo.eq_8}) for the numerical computations, and 
assumed $d_{sd} = 1$, $\alpha = 4$, $\tau = 1$, and $p_s = 0.01$ for the network parameters. For low density $\lambda$, large $p$ guarantees higher $U$, whereas, for larger $\lambda$, 
the optimal $p$ depends on the node density itself; for example, $p^{\ast} = 0.5$ for $\lambda = 0.35$.
\figref{fig.num2} compares the mean local delay $D(p)$ versus transmission probability $p$ for different node densities $\lambda$.
We observe that the optimal $p^{\ast}$ that minimizes $D(p)$ tends to decrease as the node density increases.
As the node density $\lambda$ increases, the aggregate interference at the destination generally increases. 
In order to minimize the local delay or maximize the utility, 
we need to counteract the high interference by using lower transmission probability $p$.

\figref{fig.num3} compares the ASE $A_e$ versus the mean local delay $D(p)$ while changing the node density from $\lambda = 10^{-5}$ 
(point on the lower left) to $\lambda = 10^{-1}$ (point on the upper right). 
The two dotted lines in the bottom present the conventional ALOHA system with fixed transmission probabilities for the whole range of $\lambda$, 
\textit{i.e.}, $p = 0.6$ and $p = 0.4$ for the unmarked and marked curve, respectively. 
For each $\lambda$ value, we calculated the corresponding ASE and mean local delay on a fixed $p$ by using (\ref{neo.eq_6}) and (\ref{neo.eq_7}). 
The conventional system is used as a performance benchmark for the proposed adaptive ALOHA protocol that jointly optimize ASE and delay 
for a given network parameters. For the solid line, we used numerical search to find the optimal transmission probability $p^{\ast}$ 
that satisfies (\ref{gene.neoeq_03}) on each $\lambda$ value. Based on this optimal combination $(p^{\ast}, \lambda)$, we evaluated 
the corresponding ASE and $D(p)$ by using (\ref{neo.eq_6}) and (\ref{neo.eq_7}). 
Hence, the main difference between the solid and dotted curve is that the former adaptively select the optimal $p^{\ast}$ for each $\lambda$ value, 
whereas the latter use a fixed $p$ regardless of $\lambda$. We note that the solid line with optimal $p^{\ast}$ achieves significant performance gain over the fixed $p$ case, both in terms of ASE and delay. For a fixed $A_e = 0.02$, the solid line achieves $34.7\%$ gain in terms of the mean local delay, whereas, for a given delay $D(p) = 1.8$, the solid line obtains $182\%$ gain in terms of the area spectral efficiency compared to the $p = 0.6$ case.

\section{Conclusion}
In this paper, a new metric for ASE of PPP interfered Ad-hoc
networks has been introduced to evaluate the performance of the co-channel transmission based communication systems. To jointly optimize the ASE and the mean local delay, 
we introduced a utility function, $U = A_e/D(p)$. Based on this utility $U$, the optimal ALOHA transmission probability $p$ and
the SIR threshold $\tau$ have been determined. Finally, simulation results confirmed that the joint optimization based on the $U$ metric achieves a significant performance 
gain compared to conventional communication systems.

 \begin{figure}[!t]
    \centering
        \includegraphics[width=0.85\linewidth, height=0.6\linewidth]{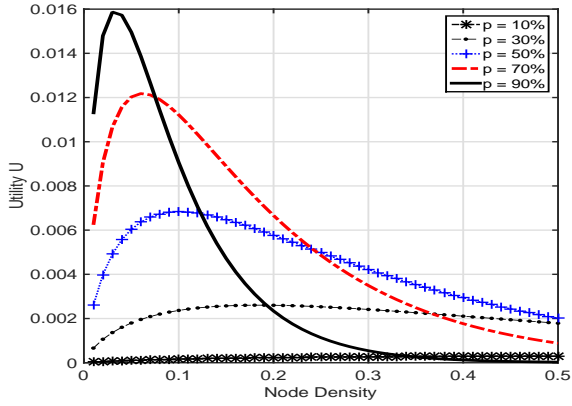}
    \caption{Utility $U$ versus node density $\lambda$ for several transmission probability $p$.}
		\label{fig.num1}
\end{figure}

 \begin{figure}[!t]
    \centering
        \includegraphics[width=0.85\linewidth, height=0.6\linewidth]{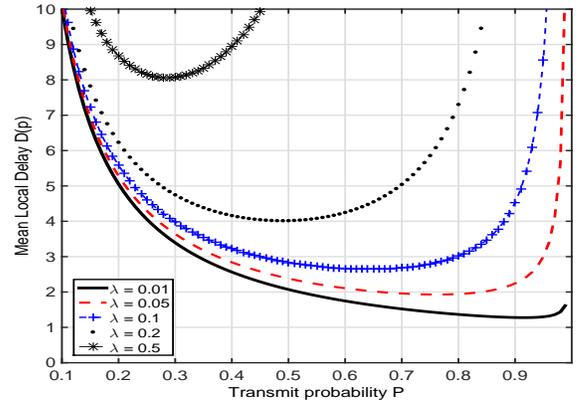}
    \caption{Mean local delay $D(p)$ versus transmission probability $p$ for different node density $\lambda$.}
		\label{fig.num2}
\end{figure}

 \begin{figure}[!t]
    \centering
        \includegraphics[width=0.85\linewidth, height=0.6\linewidth]{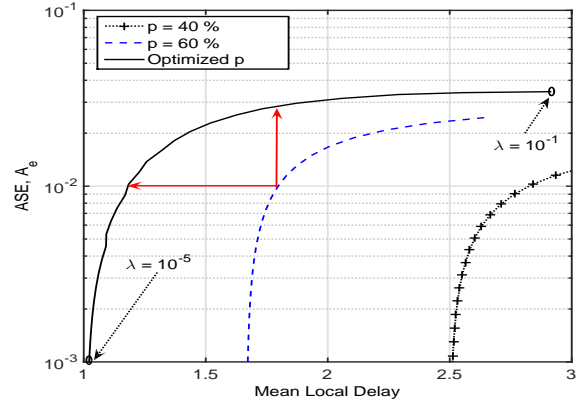}
    \caption{ASE $A_e$ versus the mean local delay $D(p)$ for node density between $10^{-5} \leq \lambda \leq 10^{-1}$.}
		\label{fig.num3}
\end{figure}

{\small \vspace{-0.1in}
\bibliographystyle{IEEEtran}
\bibliography{ref_1430}}

\end{document}